\documentclass{article}

\def\paperversionpreprint{preprint}

\def\paperversion{preprint} %

\usepackage{amsmath}
\usepackage{amssymb}
\usepackage{graphicx}
\usepackage{amsthm}
\usepackage{xspace}
\usepackage{booktabs}

\PassOptionsToPackage{numbers, sort&compress}{natbib}

\usepackage[\paperversion]{neurips_2026}

\usepackage[utf8]{inputenc} %
\usepackage[T1]{fontenc}    %
\usepackage{hyperref}       %
\usepackage{url}            %
\usepackage{booktabs}       %
\usepackage{amsfonts}       %
\usepackage{nicefrac}       %
\usepackage{microtype}      %
\usepackage{xcolor}         %

\usepackage{cleveref}

\title{PolarNet: Single-Minima Neural Network for Modeling Lyapunov Functions}

\author{%
  Yuan Zhong \\
  Faculty of Science and Technology \\
  University of Macau\\
  \texttt{mc45120@um.edu.mo} \\
  \And
  Jiaxin Cheng \\
  Faculty of Science and Technology \\
  University of Macau\\
  \texttt{yc47434@um.edu.mo} \\
  \AND
  Hefu Ye \\
  Faculty of Science and Technology \\
  University of Macau \\
  \texttt{hefuye@um.edu.mo} \\
  \And
  Yicong Zhou \\
  Faculty of Science and Technology \\
  University of Macau\\
  \texttt{yicongzhou@um.edu.mo} \\
}

\Crefname{figure}{Fig.}{Figs.}
\crefformat{equation}{#2(#1)#3}

\newcommand{\Real}{\mathbb{R}}

\newcommand{\netu}{u_\theta}
\newcommand{\netV}{V_\theta}
\newcommand{\dotnetV}{\dot{V}_{\theta}}

\newcommand{\dfdt}[1]{\frac{\mathrm{d}#1}{\mathrm{d}t}}

\newcommand{\conddecnet}{\dotnetV(x) < 0,\forall x\neq 0}

\newcommand{\Vtar}{V}

\newcommand{\statespace}{\mathcal{X}}
\newcommand{\uspace}{\mathcal{U}}
\newcommand{\neighbourhood}{\mathcal{U}}

\newcommand{\fnwarp}{\Psi}
\newcommand{\fnarch}[1]{\|\fnwarp(#1)\|^2}

\newcommand{\cpflow}{h_{cf}}

\newcommand{\properV}{single pole function\xspace}
\newcommand{\ProperV}{Single pole function\xspace}

\newcommand{\properVs}{single pole functions\xspace}

\newcommand{\ourarch}{PolarNet\xspace}

\newcommand{\continuous}[1]{\mathcal{C}^{#1}}

\newcommand{\cpfxl}{y_{\le k}}
\newcommand{\cpfxh}{y_{> k}}

\begin{document}

\maketitle

\begin{abstract}
Learning control strategies with provable stability guarantees continues to be a challenging problem. In this work, we examine a family of training-time behaviors exhibited by existing neural Lyapunov control methods under specific conditions, which can hinder the synthesis of a provably stable controller. We identify the root cause as the lack of neural network architectural guarantees on the learned Lyapunov function, and propose PolarNet, a network architecture that provably addresses these issues by structurally guarantee to have a single critical point. We provide theoretical guarantee regarding the properness and universality of PolarNet for modeling Lyapunov functions, and show that using it as a drop-in replacement in existing neural Lyapunov control methods can effectively circumvent particular difficulties in training. We conduct a set of numerical experiments to verify that PolarNet consistently maintains a single critical point and, when used as a drop-in replacement in existing neural Lyapunov control methods, successfully avoids training failures caused by the lack of architectural guarantees.
\ifx\paperversion\paperversionpreprint
The code of this paper is available at \url{https://github.com/23-zy/PolarNet}.
\fi
\end{abstract}

\section{Introduction}

Lyapunov functions are a cornerstone of modern control theory, serving as the primary tool for analyzing the stability of dynamic systems modeled in state space. Despite their near-dominant role, the search for a Lyapunov function has, for over a century, relied heavily on the physical insights and mathematical ingenuity of domain experts. Constructing such functions remains a longstanding open problem in both control theory and mathematics \cite{NEURIPS2024_aa280e73}. Although various analytical methods have been developed within the control community to assist in constructing Lyapunov functions, such as solving the algebraic Riccati equation \cite{Lewis2012} and sum-of-squares programming \cite{2002SOSLyapunov}, none of these approaches are universally applicable. Even for relatively simple systems, these methods may fail to produce a valid Lyapunov function or may incorrectly conclude that no such function exists.
Recent research (e.g., \cite{NLC,Wei2023NeuralLC,rnc-nyu-2025,NEURIPS2022_bba3160c,NEURIPS2023_disc_NLC,mehrjou2020neural}) has explored an intriguing learning-based alternative that automatically synthesizes provably stable controllers by jointly learning a control strategy and a Lyapunov function, with the latter serving as the stability certificate. Specifically, by embedding the Lyapunov conditions into the learning objective, these methods, collectively known as neural Lyapunov control, are applicable to complex uncertain nonlinear systems without sacrificing a theoretical stability guarantee. This line of work has achieved remarkable empirical success on a wide range of challenging control tasks, demonstrating a promising direction beyond conventional analytical approaches.

However, the reliance on flexible neural networks to learn Lyapunov functions introduces a new challenge that is absent in traditional analytical methods. We identify a family of training-time behaviors that can arise under specific conditions in existing neural Lyapunov control methods, that can potentially prevent the learning process from properly converging and thus blocking the synthesis of a provably stable controller. We trace this issue to the absence of a suitable structural guarantee in the neural networks used to learn Lyapunov functions. Specifically, we show that for a certain class of general nonlinear systems, any valid Lyapunov function must possess exactly one critical point, yet the network architectures commonly adopted in existing methods do not enforce this property structurally.  This issue is specific to learning-based paradigms and has therefore been overlooked in the recent shift from analytical methods to learning-based controller synthesis.

To fill this gap, we propose \ourarch, a neural network architecture that provides a rigorous structural guarantee by exploiting an interesting mathematical fact. A key theoretical result is that \ourarch always possesses exactly one critical point, thereby circumventing the training difficulty described above. Furthermore, we show that under a mild condition, \ourarch is capable of constructing Lyapunov functions for any asymptotically stable equilibrium of general dynamic systems. We validate the effectiveness of \ourarch through numerical experiments on both synthetic function approximation tasks and a challenging nonlinear control problem. The results demonstrate that \ourarch consistently maintains a single critical point and, when used as a drop-in replacement in existing neural Lyapunov control methods, successfully avoids training failures caused by the lack of architectural guarantees. To the best of our knowledge, \ourarch is the first architecture that structurally enforces the single-critical-point property for Lyapunov function learning, offering a principled and practical solution to a previously overlooked failure mode in neural Lyapunov control.

\section{Preliminary}

\theoremstyle{definition}
\newtheorem{theorem}{Theorem}
\newtheorem{lemma}{Lemma}
\newtheorem{corollary}{Corollary}
\newtheorem{condition}{Condition}
\newtheorem{definition}{Definition}
\newtheorem{remark}{Remark}

\renewenvironment{proof}[1][\proofname]{\par
  \pushQED{\qed}%
  \normalfont \topsep 0pt
  \trivlist
  \item[\hskip\labelsep\itshape#1{.}]\ignorespaces}{
  \popQED\endtrivlist
}

\subsection{Problem formulation}

Consider the following general dynamic system, which is an ordinary differential equation defined as:
\begin{equation}
    \label{eq:dynsys}
    \dot{x}=f(x,u),\qquad x(t_0)=x_0,
\end{equation}
where $x\in\statespace \subset  \Real^m$ is the system state,  $u\in\uspace \subset \Real^n$ is the control input, $\dot{x}$ denotes the derivative of $x$ w.r.t. $t$, and $f:\Real^m\times\Real^n \to \Real^m$ is an unknown nonlinear function that is locally Lipschitz continuous in $\statespace$. Suppose the system \eqref{eq:dynsys} is controllable, and $x_e\in \statespace$ is an equilibrium point of \eqref{eq:dynsys}, that is, $f(x_e,0)=0$. The goal of the task is to synthesize an appropriate input $u(x)$, such that the closed-loop  system  is asymptotically stable in the sense of \Cref{def:asymp_stable}.
\begin{definition}[Asymptotic stability \cite{khalil2002nonlinear}]
    \label{def:asymp_stable}
        Let $x_e$ be an equilibrium point of the closed-loop autonomous system  $\dot{x}=f(x,u(x))$.
        Then $x_e$ is said to be asymptotically stable if
        there exists $\delta > 0$ such that for every initial condition satisfying
        $\|x_0 - x_e\| < \delta$, we have $\lim_{t \to \infty} x(t) = x_e$.
\end{definition}
Lyapunov functions $\Vtar(x):\Real^m\to\Real$ are special functions whose existence imply that a dynamic system is asymptotically stable. Formally:
\begin{definition}[Lyapunov stability criteria \cite{khalil2002nonlinear}]
    \label{def:lyapunov_fn}
        Without loss of generality, let \( x = 0 \) be an equilibrium point of $\dot{x}=f(x,u(x))$. Suppose \( \Vtar : \mathbb{R}^m \to \mathbb{R} \) is a  radially unbounded $\mathcal{C}^1$ function satisfying: $i)~
        \Vtar(0) = 0 $,  $ii)~ \Vtar(x) > 0,  \forall \, x \neq 0$,  and  $
        \dot{\Vtar}(x) < 0,   \forall \, x \neq 0$,  	then \( x = 0 \) is  asymptotically stable.  In this case, $\Vtar(x)$ is called a  Lyapunov function of this system.
\end{definition}
Throughout this work, we consider the case where $x_e=0$ is the only desired stable equilibrium.

Recent research (e.g., \cite{NLC,Wei2023NeuralLC,rnc-nyu-2025}) has found that, by directly learning to satisfy the Lyapunov stability criteria, it is possible to synthesize a controller $u$ that stabilizes a general dynamic system. The general idea is to parameterize the controller and a corresponding Lyapunov function with two neural networks $\netu$ and $\netV$, and to train them simultaneously by minimizing the empirical risk of violating the Lyapunov stability criteria. This empirical risk \cite{NLC} is defined as:
\begin{equation}
\label{eq:lyapunov_risk_canon}
    \frac{1}{N}\sum_{i=1}^{N}\left(
        \max(0,-\netV(x_i))
        + \max(0,\dotnetV(x_i))
        + \netV^2(0)
    \right),
\end{equation}
where $\{x_i\in\statespace|i\in\{1,2,\dots,N\}\}$ is a set of points in $\statespace$, called the training set, and each $x_i$ is an $m$-dimensional state vector sampled from $\statespace$.
\subsection{Limitation of existing methods}
\label{sec:limit_of_exiting_methods}

For a proper Lyapunov function $\netV$ that satisfies $\conddecnet$,  its derivative along the system trajectory is given by the chain rule:
\begin{equation}
    {\dotnetV(x)} =\langle\nabla{\netV(x)}, \dot{x}\rangle, \qquad \forall x \in \statespace,
\end{equation}
where $\langle\cdot,\cdot\rangle$ denotes the inner product. If there exists a point \(x^* \neq 0\) such that $\langle\nabla{\netV(x)},x^*\rangle = 0$, then we would have \(\dotnetV(x^*) = 0\), contradicting \(\dotnetV(x) < 0\) for all \(x \neq 0\). Therefore, such a point cannot exist. This gives the following necessary condition for valid Lyapunov function $\netV$:
\begin{equation}
    \label{eq_no_nz_crit_point}
    \nabla{\netV(x)} \neq 0,\qquad\forall x\neq 0.
\end{equation}
Recent works of neural Lyapunov control (e.g., \cite{rnc-nyu-2025,lnet}) typically model Lyapunov functions with specialized networks $\netV(x): \Real^m \to \Real$ that are always positive definite, and pursue the $\conddecnet$ condition by minimizing its empirical risk function as defined in \cref{eq:lyapunov_risk_canon}. With such network that provides structural guarantee of $\netV(0)=0$ and $\netV(x)>0,\forall x\neq 0$, \cref{eq:lyapunov_risk_canon} can be reduced to:
\begin{equation}
\label{eq:lyapunov_risk_reduced}
    \frac{1}{N}\sum_{i=1}^{N}\max\left(0,\dotnetV(x_i)\right).
\end{equation}
While these methods have demonstrated the ability to synthesize controllers that can be effective in practice, their network architectures do not inherently guarantee that \cref{eq_no_nz_crit_point} stays true. As a result, training may encounter difficulties in certain scenarios. Consider, for example, a setup where we train $\netu$ and $\netV$ by minimizing the empirical risk in \cref{eq:lyapunov_risk_reduced} over  uniformly sampled $x_i \in \statespace$ with scalar state and input signals. Under this training strategy, the situation depicted in \Cref{fig:fail_eg_for_explain} may occur: at some point during training, $\netV$ learned a shape with multiple local minima (see \Cref{fig:fail_eg_for_explain}(a)). \Cref{fig:fail_eg_for_explain}(b) shows the partial derivative
$\partial\netV/\partial x$ of this learned function, where the existence of points with zero derivative (other than the origin) clearly violates the necessary condition (4). In addition, since \cref{eq:lyapunov_risk_reduced} is computed with the current $\netV$ and $\netu$ that are still being trained, this loss term is essentially encouraging $\netu$ to steer the states along the descending directions of $\netV$, regardless of whether $\netV$ itself is a valid Lyapunov function or not. Consequently, the closed-loop dynamics determined by $f$ and $\netu$  may take a shape similar to that shown in \Cref{fig:fail_eg_for_explain}(c). In this case, $\dotnetV<0$ is not always true, and $\dotnetV=0$ occurs at the origin, $P$ and $Q$ (see \Cref{fig:fail_eg_for_explain}(d)). However, the empirical risk in \cref{eq:lyapunov_risk_reduced} remains zero, even if the Lyapunov stability criteria is clearly violated at points $P$ and $Q$. Furthermore, training of $\netV$ and $\netu$ is performed on a finite number of samples $\{x_i\}$ in practice. So even when the roots of $\dot{x}=0$ do not perfectly coincide with $P$ and $Q$, the empirical risk can still vanish to zero at this stage, providing little useful gradient information for continued learning. Consequently, $\netu$ may fail to eliminate the spurious equilibria (e.g., $P$ and $Q$), and more importantly, the training may stall at these points, offering no guarantee that the learned controller achieves asymptotic stability of the desired equilibrium $x=0$.

\begin{figure}[!ht]
    \centering
    \includegraphics[width=1\linewidth]{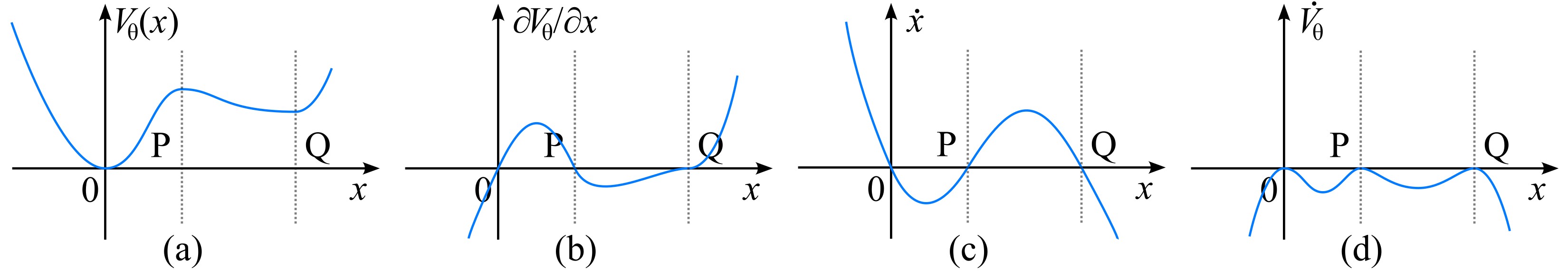}
    \caption{A possible scenario that may arise when using existing neural network architectures (e.g., \cite{NLC,Wei2023NeuralLC,rnc-nyu-2025,lnet}), shown as functions of the state $x$ along the horizontal axis. As a numerical example, consider a scalar system \(\dot{x}=u - x(x-P)(x-Q)\). The learned Lyapunov candidate function could take the form \(\netV(x)=x^2+\sin^2(\pi x)\), and $P,Q$ are the smallest two positive roots of equation $\partial{\netV}/\partial x=0$. If at \(x=P\) (or \(x=Q\)), the learned controller satisfies \(u_\theta = 0\), then $\dot{x}=0$ and hence the system trajectory will remain at \(x(t)=P\) (or \(x(t)=Q\)), and the true equilibrium \(x=0\) is not globally asymptotically stable, which hinders the successful synthesis of a provably stable controller. An experimental demonstration of such a scenario is provided in \Cref{sec:exp:controller_synthesis}.}
\label{fig:fail_eg_for_explain}
\end{figure}

Existing methods adopt various strategies to mitigate this issue. One common approach adds a small positive constant $\epsilon$ to the loss term in \cref{eq:lyapunov_risk_reduced}, making it non-zero unless $ \dotnetV < -\epsilon$, thereby expanding the region where the loss provides useful gradients \cite{Wei2023NeuralLC,rnc-nyu-2025}. Another approach actively seeks out points $x_i \neq 0$ where $\dot{V}_\theta(x_i) \ge 0$ via numerical optimization and re-construct the training set to include these points, allowing the model to focus on eliminating these violations \cite{NLC}. Similarly, a third strategy simulates closed-loop trajectories to assess the stability of initial states, constructing the training set around the current region of attraction (RoA) to maintain and gradually expand it during training \cite{Wei2023NeuralLC,mehrjou2020neural}. While these strategies may alleviate the issue during training, they do not completely eliminate it. We demonstrate this with an experiment in \Cref{sec:exp:controller_synthesis}. To the best of our knowledge, no existing method fundamentally resolves this issue by enforcing a network structural guarantee on the learned Lyapunov function itself.

It should also be noted that this issue is unique to learning-based methods that use flexible neural networks to model Lyapunov function, and is not present in analytic methods \cite{khalil2002nonlinear}.

\section{Method}

\subsection{Network architecture design}

To address the issue described in \Cref{sec:limit_of_exiting_methods}, we propose \ourarch, a network architecture that guarantees a single-critical-point shape in addition to positive-definiteness. By construction, this architecture strictly enforces \cref{eq_no_nz_crit_point} to be true, regardless of how its parameters are set. Formally speaking, \ourarch models a \properV, which we define as:
\begin{definition}[\ProperV]
\label{def:properV}
    A \properV is a function $\Vtar(x):\statespace\to\Real,\statespace\subset\Real^m,0\in\statespace$, such that it:

    $i)~$ is $\continuous{\infty}$ smooth $\forall x\in\statespace$;

    $ii)~$ is {positive definite}: $\Vtar(0)=0$; $\Vtar(x)>0,\forall x\neq 0$;

    $iii)~$ has exactly one critical point at $x=0$: $\nabla{\Vtar}\neq 0,\forall x\neq 0$;

    $iv)~$ is {proper} on $\statespace$, i.e., every sublevel set $\{x\in\statespace|V(x)\le c\}$ of $\Vtar$ is {compact}.\footnote{When $x\in\Real^m$ and $\Vtar$ is continuous, $\Vtar(x)$ being proper is equivalent to $\Vtar(x)$ being radially unbounded.}
\end{definition}
We additionally show in \Cref{sec:theoretical_guarantee} that, there exists a Lyapunov function, which is also a \properV, for the equilibrium of the general dynamic system \cref{eq:dynsys}.

We first reveal an interesting connection between a \properV and a smoothly invertible mapping $\fnwarp(x):\Real^m\to\Real^m$. Exploiting this connection, by integrating invertible neural networks (i.e., neural networks that are invertible by design, e.g., \cite{dinh2017density,NEURIPS2018_d139db6a,NIPS2016_ddeebdee,NIPS2017_6c1da886}) into the model architecture, we can naturally construct a neural network that guarantees to be a \properV.

We prove that, under a mild condition, any \properV can be expressed in the form $\Vtar=\fnarch{x}$, where $\fnwarp(x):\Real^m\to\Real^m$ is a $\continuous{\infty}$ smooth, invertible and zero-crossing function, and $\|\cdot\|$ is the standard Euclidean norm. Geometrically, this means $\Vtar$ can always be smoothly deformed to match an Euclidean ball $\|\cdot\|^2$, as illustrated in \Cref{fig:design_intuition}:

\begin{figure}[!ht]
    \centering
    \includegraphics[width=1\linewidth]{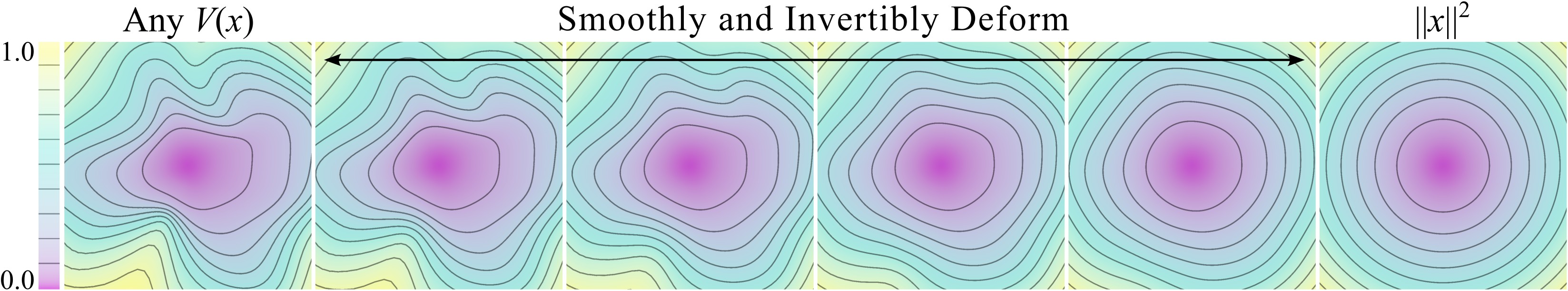}
    \caption{Graphical illustration of \Cref{theorem:all-can-be-repr}: The case of $\Vtar(x):\Real^m\to\Real, m=2$. As illustrated, a \properV $\Vtar(x)$ can be spatially distorted by a smoothly invertible mapping $\fnwarp(x):\Real^m\to\Real^m$, such that it matches an Euclidean ball $\|x\|$ after distortion. In other words, $\Vtar(x)$ can always be reformulated as the squared Euclidean norm $\|z\|^2$ in a smoothly distorted space $z=\fnwarp(x)$. For visual clarity, we visualize $\Real^2\to\Real$ functions as contour lines.}
\label{fig:design_intuition}
\end{figure}

\begin{theorem}
\label{theorem:all-can-be-repr}
    For all \properVs $\Vtar(x):\Real^m\to\Real$, if the origin is a non-degenerate critical point of $\Vtar$ (i.e., the Hessian matrix $D^2\Vtar(0)$ is invertible), then there exists a $\continuous{\infty}$ smooth, invertible and zero-crossing vector-valued function $\fnwarp(x):\Real^m\to\Real^m$, such that $\Vtar=\fnarch{x}$ holds throughout $\statespace$.
\end{theorem}
\begin{proof}
\label{proof:all_eligible_fn_can_be_written_in_gnf}
    The idea of the proof is to directly construct a diffeomorphism $\fnwarp$ between the sublevel sets $\{x\in\statespace|V(x)\le c\}$ and standard spheres. This is done by smoothly identifying level sets $L_c=\{x\in\statespace|V(x)=c\}$ to each other with the flow \cite{Lee2012} of $\Vtar(x)$, thus constructing a diffeomorphism between the product structure $L_c\times(0,\infty)$ and $\statespace\setminus\{0\}$. Then we use the Morse lemma \cite{Milnor1963,Nirenberg2001} to obtain a diffeomorphism between a particular level set of $\Vtar(x)$ and a standard sphere. By composing these two diffeomorphism, we can construct an eligible $\fnwarp$ such that $\Vtar(x)=\fnarch{x}$. For a detailed and rigorous proof, see \Cref{apdx:proof_all_properv_can_be_written_in_gnf}.
\end{proof}
\begin{remark}
    While \Cref{theorem:all-can-be-repr} requires a non-degenerate critical point, this is a mild condition, as non-degenerancy is a common situation, thus any degenerate critical point can be perturbed by an arbitrarily small perturbation to become a non-degenerate critical point (see Chapter 1.7 of \cite{Guillemin2010}).
\end{remark}
We construct \ourarch $\netV$ exactly as $\netV(x) = \fnarch{x}$. We use the affine coupling flow architecture \cite{dinh2017density,NEURIPS2018_d139db6a,NIPS2016_ddeebdee,NIPS2017_6c1da886} to implement $\fnwarp$.
This architecture $\cpflow(x):\Real^m\to\Real^m$ is defined as:
\begin{subequations}
\label{eq:coupling_flow}
    \begin{align}
        h^{(l)}\begin{pmatrix}
            \cpfxl \\
            \cpfxh
        \end{pmatrix} &=
        \begin{pmatrix}
            \cpfxl \\
            \exp({f_s^{(l)}(\cpfxl)}) \odot \cpfxh + f_t^{(l)}(\cpfxl)
        \end{pmatrix},l\in\{1,2,\dots,L\},y=h^{(l-1)}, \\
        h^{(0)}(x) &= x, \\
        \cpflow(x)&=h^{(L)} \circ h^{(L-1)}\circ \space \cdots \space \circ h^{(1)}(x),
    \end{align}
\end{subequations}
where $k\in\{1,\dots,m-1\}$, $\odot$ is the Hadamard product (i.e., element-wise product), $\exp$ is applied in an element-wise manner, and $\circ$ is the function composition operator $a\circ b(y)=a(b(y))$, $a\circ b\circ c = a\circ(b\circ c)$. Both $f_s$ and $f_t$ are arbitrary neural networks. For each layer $h^{(l)}$, its input $y=(y_1,\cdots,y_m)^\intercal$ is partitioned into 2 parts $\cpfxl=(y_1,\cdots,y_k)^\intercal$ and $\cpfxh=(y_{k+1},\cdots,y_m)^\intercal$. After this layer, $\cpfxl$ is kept unchanged, and $\cpfxh$ is being transformed in an affine manner, with scale $\exp{f_s(\cpfxl)}$ and translation $f_t(\cpfxl)$ derived from $\cpfxl$. Both $f_s(\cdot)$ and $f_t(\cdot)$ are arbitrary neural networks. This architecture guarantees that, if $f_s(\cdot)$ is finite, the inverse of $h^{(l)}$ exists, and can be formulated as follows:
\begin{equation}
    h^{-1(l)}\begin{pmatrix}
        \cpfxl \\
        \cpfxh
    \end{pmatrix} =
    \begin{pmatrix}
        \cpfxl \\
        \\
        \dfrac{\cpfxh - f_t(\cpfxl)}{\exp{f_s(\cpfxl)}}
    \end{pmatrix}.
\end{equation}
In order to ensure that the entire network $\netV$ is $\continuous{\infty}$ smooth, we only use smooth activation functions (tanh) throughout the network. To achieve $\netV(0)=0$, we additionally ensure that $\fnwarp(0)=0$. Specifically, we construct the translation function $f_t$ in each affine layer $h^{(l)}$ using a multi-layer perceptron (MLP) \cite{Rumelhart1986} without bias (and using tanh activation, which is zero-crossing). This guarantees that $\exp{f_s(0)} \odot 0 + f_t(0)=0$. In other words, each layer is zero-crossing, so does their composition $\fnwarp$.

\subsection{Theoretical guarantee}
\label{sec:theoretical_guarantee}

\Cref{theorem:all-can-be-repr} has shown that all \properVs $\Vtar(x)$ can be formulated as $\Vtar(x)=\fnarch{x},\forall x\in\statespace\setminus\{0\}$. As has been proved in \cite{cpflow_universal_proof}, the affine coupling flow architecture is a universal approximator of $\continuous{2}$ diffeomorphisms in $\Real^m$. In other words, all $\continuous{2}$ differentiably invertible mappings in $\Real^m$ can be approximated by the affine coupling flow architecture to arbitrary precision. And by restricting all components of this architecture to be $\continuous{\infty}$, we can obtain a universal approximator for $\fnwarp$. Hence, a \properV can always be approximated by $\netV$. Now we also show that $\netV(x)=\fnarch{x}$ is always a \properV, regardless of how its parameters are set (as long as they are finite):
\begin{theorem}
\label{theorem:netv_is_always_a_properV}
    If $\fnwarp:\Real^m\to\Real^m$ is a $\continuous{\infty}$ smooth, zero-crossing invertible mapping, then $\fnarch{x}$ is a \properV.
\end{theorem}
\begin{proof}
    We prove for each property:

    $i)~$ Since both $\|\cdot\|^2$ and $\fnwarp$ are $\continuous{\infty}$ smooth, their composition $\fnarch{\cdot}$ is also $\continuous{\infty}$ smooth.

    $ii)~$ Invertibility of $\fnwarp$ implies that $\fnwarp(x)\neq\fnwarp(0),\forall x\neq0$, hence $\fnwarp(x)\neq\fnwarp(0)=0,\forall x\neq 0$. By definition of the squared norm $\|\cdot\|^2$, we have $\fnarch{x}\neq0,\forall x\neq0$. Since $\|\cdot\|^2\ge0$, $\fnarch{x}$ is positive definite.

    $iii)~$ Since $\fnwarp$ is invertible, it has no critical point. $\|\cdot\|^2$ has exactly one critical point at 0, $\fnwarp(x)=0$ only at $x=0$. By the chain rule of derivative computation, $\fnarch{x}$ has exactly one critical point at $x=0$.

    $iv)~$ $N(y)=\|y\|^2$ is proper on $\Real^m$ (with codomain $[0,+\infty)$), hence for any compact set $K\subset[0,+\infty)$, its preimage $N^{-1}(K)$ is also compact. Invertibility of $\fnwarp$ implies continuity of $\fnwarp^{-1}$, and the continuous image of any compact set is also a compact set. Thus $\fnwarp^{-1}(N^{-1}(K))$ is also compact. Hence $\fnarch{x}$ is proper on $\Real^m\supset\statespace$.

    We have thus shown that $\fnarch{x}$ is always a \properV.
\end{proof}
We also show that, for all asymptotically stable equilibrium of a general dynamic system \cref{eq:dynsys}, there exists a \properV $\Vtar$, such that $\Vtar$ is a Lyapunov function of this equilibrium:
\begin{theorem}
\label{theorem:gas_lip_sys_always_has_properV}
    Let $x=0$ be an asymptotically stable equilibrium of a dynamic system $\dot{x}=f(x,u(x))$, where $f(x,u(x))$ is locally Lipschitz. Let $\mathcal{R}_A\in\statespace$ be the region of attraction of $x=0$. Then there exists a Lyapunov function $\Vtar(x)$ for this equilibrium, that is also a \properV.
\end{theorem}
\begin{proof}
    By Kurzweil's converse Lyapunov theorem (\emph{Theorem 4.17} in \cite{khalil2002nonlinear}), and by the conclusion shown in \cref{eq_no_nz_crit_point}, for an asymptotically stable equilibrium $x=0$ of a dynamic system, there exists a Lyapunov function $\Vtar(x):\statespace\to\Real$ that is: $i)~ \continuous{\infty}$ smooth, $ii)~$ positive definite, $iii)~$ has exactly one critical point at $x=0$, and $iv)~$ satisfies $\forall c>0,\{x\in\statespace|\Vtar(x)\le c\}$ is compact, i.e., $\Vtar(x)$ is proper in $\mathcal{R}_A$. In other words, it is a \properV (By \Cref{def:properV}, letting $\statespace=\mathcal{R}_A$, one can conclude that $\Vtar$ is a \properV).
\end{proof}
\Cref{theorem:gas_lip_sys_always_has_properV} shows that for any controller synthesis task demanding an asymptotically stable equilibrium, there exists a Lyapunov function that can be accurately modeled by \ourarch. Combining with \Cref{theorem:netv_is_always_a_properV}, this provides a principled method to circumvent training difficulties caused by the lack of architectural guarantees in existing neural Lyapunov control methods.

\section{Experiments}

\subsection{Network architecture implementation}
\label{sec:exp}

In all experiments, we implement \ourarch with 4 affine coupling flow layers, (i.e., use the architecture described in \cref{eq:coupling_flow} with $L=4$), evenly partition the dimension of input vectors, and alternate the part being preserved in adjacent layers. In other words, we implement $\fnwarp$ as:
\begin{equation}
\label{eq:cpf_impl}
    \fnwarp(x) = \begin{pmatrix}
        \cpfxl \\
        g^{(4)}(\cpfxl, \cpfxh)
    \end{pmatrix} \circ \begin{pmatrix}
        g^{(3)}(\cpfxh, \cpfxl) \\
        \cpfxh
    \end{pmatrix} \circ \begin{pmatrix}
        \cpfxl \\
        g^{(2)}(\cpfxl, \cpfxh)
    \end{pmatrix} \circ \begin{pmatrix}
        g^{(1)}(\cpfxh, \cpfxl) \\
        \cpfxh
    \end{pmatrix},
\end{equation}
where $g^{(l)}(a, b) = \exp(f_s^{(l)}(a))\odot b + f_t^{(l)}(a)$ and $y^{0} = x$. All $f_s^{(l)}$ and $f_t^{(l)}$ are MLPs with two 12-dimensional hidden layers. And $f_t^{(l)}$ does not have a bias term. Each linear layer in these MLPs are followed by a tanh activation function, with the last layer (of each MLP) as an exception, which uses an identity activation.

\subsection{Studying properties of \ourarch}
\label{sec:exp:fn_fit}

In this experiment, we verify that \ourarch is always a \properV. Particularly, we show that it would never produce critical point other than 0, even when directly trained to fit a function with multiple critical points. Specifically, we manually design several scalar fields $V(x):\Real^2\to\Real, x\in(-1,+1)\times(-1,+1)$ and train a neural network $\netV$ to fit them by uniformly sample $x_i\in(-1,+1)\times(-1,+1)$, and minimize a mean-squared-error (MSE) loss function:
\begin{equation}
\label{eq:loss_mse}
    L=\frac{1}{N}\sum_{i=1}^{N}\left(\netV(x_i)-V(x_i)\right)^2.
\end{equation}
We compare \ourarch against several existing and representative network architectures for modeling Lyapunov functions. Namely:
$i)~$ MLP \cite{NLC,Rumelhart1986}: A multi-layer perceptron with three 64-dimensional hidden layers and tanh activation.
$ii)~$ Lyapunov Net \cite{lnet}: defined as $|h(x)-h(0)|+\gamma \|x\|^2$, where $h(x)$ is an MLP with three 64-dimensional hidden layers, $\gamma>0$ is a hyperparameter that is set to 1e-2, as in \cite{lnet}.
And $iii)~$ Wei et al. \cite{rnc-nyu-2025}: defined as $\frac{1}{2}\beta x^\intercal x+\frac{1}{2}\phi_\theta(x)^\intercal\phi_\theta(x)$, where $\phi_\theta(x)$ is an MLP with three 64-dimensional hidden layers, no bias term, and using tanh activation, $\beta$ is a hyperparameter that is set to 1e-6, as in \cite{rnc-nyu-2025}. The computation resource being used to conduct this experiment is described in \Cref{apdx:compres}.

We train the models for 10000 steps using the Adam optimizer \cite{Adam}, with $\beta_1=0.9,\beta_2=0.999,\epsilon=\text{1e-8}$, 5e-3 learning rate with 400 steps warm-up, and a batch size of 65536.
The results are presented in \Cref{fig:exp:fit_proper}. We found that, when $\Vtar$ is a \properV, all methods have comparable performance. However, when $\Vtar$ is not a \properV, existing architectures would still fit it without problem, thus producing multiple local minima. As a consequence, they may learn to exhibit multiple local minima when being used to model Lyapunov functions in controller synthesis tasks, thus hindering the successful synthesis of a globally stable controller. On the other hand, \ourarch remains a \properV in all cases. Hence it can always guarantee that \cref{eq_no_nz_crit_point} stays true throughout training process, avoiding this family of potential fail cases.

\begin{figure}[!ht]
    \centering
    \includegraphics[width=1\linewidth]{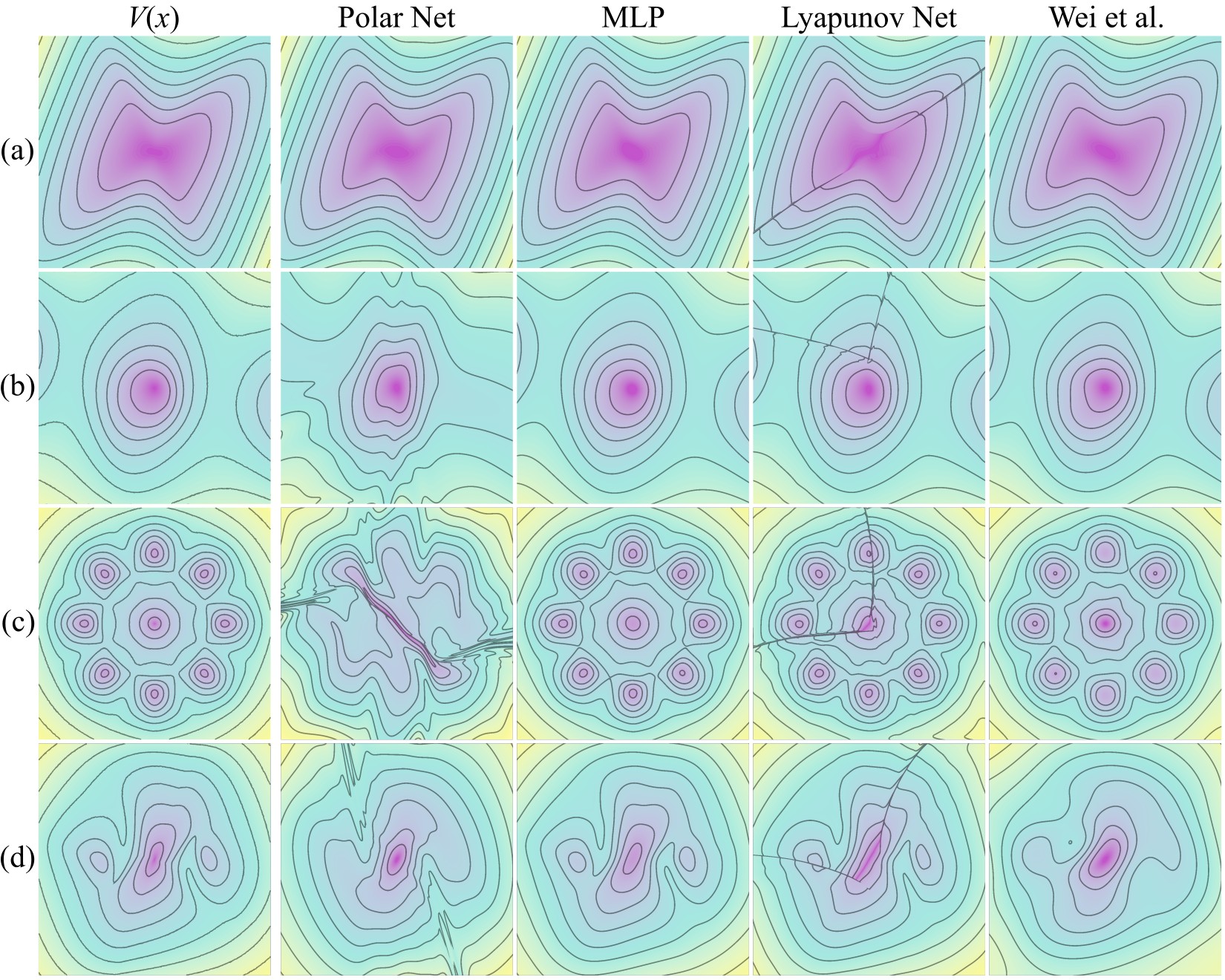}
    \caption{Results of fitting four different $\Vtar(x):\Real^2\to\Real$. In (a), $\Vtar$ is a \properV. In (b), (c) and (d), $\Vtar$ has multiple local minima and thus is not a \properV. Note that a network that guarantees to be a \properV should fail to faithfully represent $\Vtar$ in cases (b), (c), (d). We visualize $\Real^2\to\Real$ as contour lines for visual clarity, using the same scheme as in \Cref{fig:design_intuition}.}
\label{fig:exp:fit_proper}
\end{figure}

\subsection{Controller synthesis}
\label{sec:exp:controller_synthesis}

We verify the effectiveness of \ourarch in a controller synthesis task using numerical simulation. Consider the following nonlinear dynamic system:
\begin{subequations}
\label{eq:exp_sys}
    \begin{align}
        & \begin{cases}
            \dot{x}_1 = x_2 \\
            \dot{x}_2 = \dfrac{1}{2} \sin(\pi x_1) + (20\exp(-2x_1^2-\dfrac{1}{2}x_2^2)-10)u_1 + 50g(x_1)u_2 - \dfrac{1}{10}u_2
        \end{cases},\\
        g(x_1) = & \begin{cases}
            \dfrac{k}{2}\left(\dfrac{x_1}{k}\right)^2 \quad & \text{if } |x_1| < k \\
            |x_1|-\dfrac{k}{2} \quad & \text{otherwise}
        \end{cases},
    \end{align}
\end{subequations}
where $k=0.02$, $x=(x_1\;x_2)^\intercal,x\in\statespace=(-1,+1)\times(-1,+1)$ is the system state, $u=(u_1\;u_2)^\intercal\in\Real^2$ is the input. Note that $g(x_1)$ is $\continuous{1}$ and globally Lipschitz, because its two parts have their values and first order derivatives coincide at $|x_1|=k$, and the derivative is bounded between $[-1,+1]$. This system is also a control-affine system (i.e., it can be formulated as $\dot{x}=f_1(x)+f_2(x)u$, this matches the assumption made in \cite{Wei2023NeuralLC}). Also note that this system model is unknown to the controller.

We verify the effectiveness of \ourarch with a recent neural Lyapunov control method proposed in \cite{Wei2023NeuralLC}.
We first show that, when using an existing neural network architecture as in \cite{Wei2023NeuralLC}, this method could fail to stabilize the dynamic system \cref{eq:exp_sys}. And when we duplicate the same setting, except that we replace the network for Lyapunov function by \ourarch, we can successfully stabilize the dynamic system. We use a learning rate of 5e-6 for both the controller and Lyapunov function networks, and use a batch size of 32. Training is performed for 200 iterations (100 gradient update steps for each iteration). Further technical details used in this experiment are given in \Cref{apdx:tech_detail_of_exp}. The baseline architecture for comparison is implemented exactly as described in \cite{Wei2023NeuralLC}. The results are as shown in \Cref{fig:exp_fin_all}.

\begin{figure}[!ht]
    \centering
    \includegraphics[width=1\linewidth]{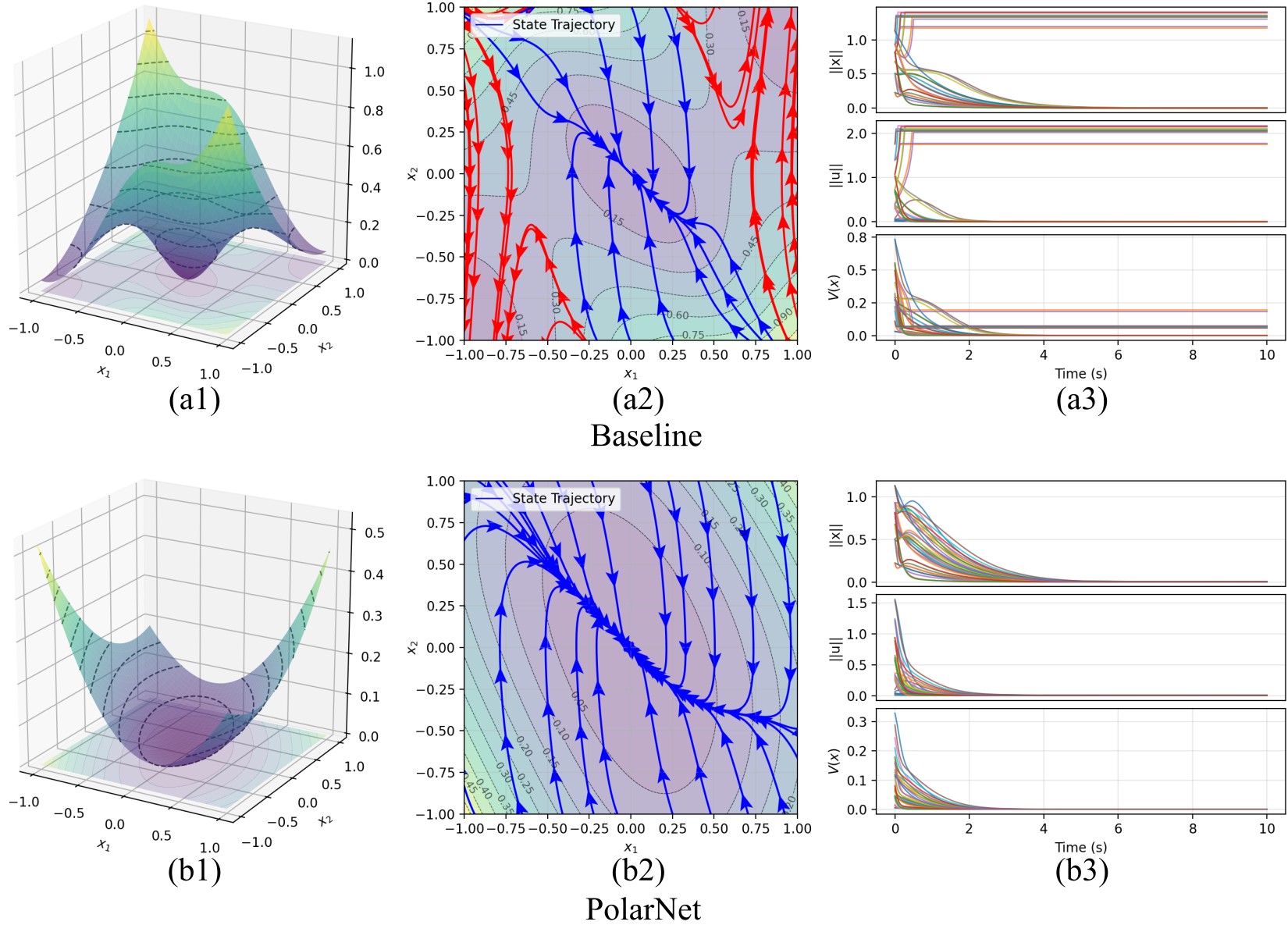}
    \caption{Controller synthesis result with different network architectures: (a1), (b1): Lyapunov function. (a2), (b2): Phase portrait, unstable trajectories are marked in red. (a3), (b3): 36 numerically simulated trajectories. (a): Using the network architecture as in \cite{Wei2023NeuralLC}. (b): Using \ourarch. $\Vtar(x)$ and $\|x\|$ trajectories are collected by uniformly sample $x_0\in[-0.8,+0.8]^2$, and if the state leaves the valid region $\statespace=(-1,+1)^2$ of the system dynamics, simulation is terminated for that trajectory. Notice how states evolve along descending directions of $\netV$ for $x\in(-1,+1)^2$ in (a) but do not converge to the origin, as analyzed in \Cref{sec:limit_of_exiting_methods}.}
    \label{fig:exp_fin_all}
\end{figure}

As shown in \Cref{fig:exp_fin_all}, in this experimental setup, existing neural network architecture, that does not guarantee \cref{eq_no_nz_crit_point}, may learn a shape with multiple local minima, causing controller synthesis to fail. When directly replacing the network with \ourarch, $\netV$ is guaranteed to have a unique minima at 0, and a controller that stabilize the dynamic system in the entire $\statespace$ can be successfully synthesized. We also conduct experiments with another method \cite{rnc-nyu-2025} and obtained similar results. Details can be found in \Cref{apdx:exp2025}.

It should be noted that this experiment is designed to verify both the existence of the issue described in \Cref{sec:limit_of_exiting_methods} and the effectiveness of \ourarch in circumventing it. Our method is at the principle level and thus applies universally to any system, whether synthetic or real world, provided the system satisfies the underlying mathematical assumptions.

\section{Limitations and future directions}
\label{sec:limitations}

The theoretical guarantee provided by our work assumes a locally Lipschitz dynamic system. When the system is not Lipschitz, \ourarch may still be applied, but does not guarantee in general that it can always represent a valid Lyapunov function for the system (i.e., \Cref{theorem:gas_lip_sys_always_has_properV} does not hold in this case). Future works may seek to generalize \ourarch to work with a broader continuum of dynamic systems.

Our method also only focus on continuous time dynamic system and stability in the Lyapunov sense, while discrete systems and other stability definitions are also of interest in many cases. Our work does not consider the case in discrete systems and stability definitions other than the Lyapunov stability.

It should also be noted that \ourarch is specifically designed for controller synthesis task with a single desired equilibrium. For out-of-scope tasks such as finding stability certificate of existing dynamic systems \cite{pmlr-v155-boffi21a}, \ourarch does not have clear advantages over existing architectures, and may even be less expressive if the intention is to find multiple stable equilibriums. Care should be taken when trying to apply \ourarch in any of these out-of-scope tasks.

\section{Related works}

Neural Lyapunov control methods follow a common scheme: they simultaneously learn a controller $\netu$ and its Lyapunov function $\netV$, thereby synthesizing a controller with provable stability guarantees. Several works (\cite{NLC,NEURIPS2022_bba3160c,NEURIPS2023_disc_NLC}) adopt a learner-falsifier architecture, where a falsifier actively searches the state space for points where Lyapunov stability criteria are violated. These points, termed counterexamples, are added to the training set for learning $\netu$ and $\netV$. The learner then minimizes the empirical risk of Lyapunov stability criteria violation on this training set by using gradient descent to update $\netu$ and $\netV$. When no counterexample can be found, the synthesized controller $\netu$ is considered stable, and $\netV$ is a valid Lyapunov function certifying its stability.
Among these methods, \cite{NLC,NEURIPS2022_bba3160c} formulate the counterexample generation problem as a \emph{satisfiability modulo theories} (SMT) \cite{barrett2018satisfiability} problem and solve it using numerical SMT solvers.
In contrast, \cite{NEURIPS2023_disc_NLC} propose a \emph{projected gradient descent} (PGD) algorithm to approximately find counterexamples, improving the training efficiency.

Some other methods, such as \cite{Wei2023NeuralLC,mehrjou2020neural}, deviate from this learner-falsifier architecture. They iteratively evaluate the region of attraction (RoA) of $\netu$ during training and construct training sets that cover the current RoA and a surrounding region, thereby maintaining a contiguous RoA and gradually expanding it. Among these, \cite{mehrjou2020neural} determines the RoA solely by checking the Lyapunov stability criteria with the current $\netu$ and $\netV$, whereas \cite{Wei2023NeuralLC} additionally collects numerically simulated state trajectories to assess the stability of initial states more reliably. Most of these methods target continuous-time dynamical systems, while \cite{NEURIPS2023_disc_NLC} apply neural Lyapunov control to discrete-time systems. \cite{rnc-nyu-2025} propose a different approach that trains $\netu$ and $\netV$ on a finite grid over the state space without constructing more sophisticated training sets, and report empirical success. \cite{Wei2023NeuralLC,rnc-nyu-2025,lnet} use neural networks that are positive definite by construction to model $\netV$, eliminating the need for an additional loss term that enforces positive definiteness. To the best of our knowledge, \ourarch is the first architecture that structurally enforces the single-critical-point property for Lyapunov function learning.

\section{Conclusion}

In this work, we analyzed the cause of a family of training-time difficulties that may present in existing neural Lyapunov control methods, identifying their root cause as lack of neural network architectural guarantees, and propose a new network architecture, \ourarch, that addresses this issue by construction. We derived theoretical guarantees that \ourarch possess the desired mathematical properties by construction, and can be applied to all asymptotically stable equilibrium of locally Lipschitz dynamic systems. We then verified the effectiveness of \ourarch with numerical experiments, demonstrating that \ourarch indeed possesses the desired properties in practice, and can circumvent certain difficulties encountered in the training process of existing neural Lyapunov control method, as we have analyzed. We also noted on the theoretical assumptions we made. Future works may seek to generalize \ourarch to be provably effective in a wider range of scenarios.

\bibliographystyle{abbrv}
\bibliography{references.bib}

\appendix

\section{Proof of \Cref{theorem:all-can-be-repr}}
\label{apdx:proof_all_properv_can_be_written_in_gnf}

\begin{proof}
The proof goes by first showing that we can smoothly identify different regular level sets (i.e., level sets that do not contain the critical point 0) of $\Vtar$, and construct a diffeomorphism $\Phi : L_{c} \times (-c, \infty) \to \statespace\setminus\{0\}$. We then smoothly map a regular level set to the standard sphere using Morse lemma, thus obtaining $\fnwarp_0$ such that $\Vtar(x) = \|\fnwarp_{0}(x)\|^{2},\forall x\in\statespace\setminus\{0\}$. Finally we extend $\fnwarp_0$ smoothly to 0 to obtain a suitable $\fnwarp$.

\subsection{Gradient flow and a global product structure}
\label{prfsec:global_product_struct}

Since \(\Vtar\) is \(\continuous{\infty}\) and has no critical points in \(\statespace\setminus\{0\}\), its gradient \(\nabla\Vtar\) is also \(\continuous{\infty}\) and never zero on this set.
Hence the \emph{normalised gradient vector field}
\begin{equation}
    X_{\mathrm{grad}}(x) = \frac{\nabla \Vtar(x)}{\|\nabla \Vtar(x)\|^{2}} , \qquad x \in \statespace\setminus\{0\},
\end{equation}
is well defined and of class \(\continuous{\infty}\).
Let \(\phi_{t}(x)\) denote the \emph{flow} of \(X_{\mathrm{grad}}\), i.e.\ the unique maximal solution of the ordinary differential equation \(\frac{d}{dt}\phi_{t}(x) = X_{\mathrm{grad}}(\phi_{t}(x))\) with \(\phi_{0}(x)=x\). For each $x$, the flow is a function of $t$ defined on a maximal open interval
\((T_{-}(x),\, T_{+}(x))\) containing \(0\) (by the \emph{Fundamental Theorem on Flows} in \cite{Lee2012}).
If \(T_{+}(x)<\infty\), then the trajectory must leave every compact subset of \(\statespace\) as \(t\to T_{+}(x)\) (by the \emph{Escape Lemma} in \cite{Lee2012}).

Now compute the derivative of \(\Vtar\) along the flow:
\begin{equation}
    \dfdt{}\,\Vtar(\phi_{t}(x)) = \bigl\langle \nabla \Vtar(\phi_{t}(x)),\, \dot\phi_{t}(x) \bigr\rangle
    = \bigl\langle \nabla \Vtar(\phi_{t}(x)),\, \frac{\nabla \Vtar(\phi_{t}(x))}{\|\nabla \Vtar(\phi_{t}(x))\|^{2}} \bigr\rangle = 1.
\end{equation}
Therefore, \(\Vtar(\phi_{t}(x)) = \Vtar(x) + t\) as long as the flow is defined.
Because \(\Vtar\) is proper on \(\statespace\) and vanishes only at the origin,
by definition of properness, the sublevel set \(S_{c} = \{x\in\statespace : \Vtar(x) \le c\}\) is compact for every \(c\ge0\).

For any \(T>0\) and any starting point \(x\neq0\), the forward segment
\(\{\phi_{t}(x) : 0 \le t \le T\}\) is contained in the compact sublevel set \(S_{\Vtar(x)+T}\),
which lies inside \(\statespace\setminus\{0\}\) because \(\Vtar\ge \Vtar(x)>0\) there.
The Escape Lemma then implies that the flow cannot grow to infinity in finite forward time, hence
\(T_{+}(x)=+\infty\).
Similarly, moving backwards, \(\Vtar\) decreases linearly until it reaches \(0\). The trajectory approaches the origin but never escapes \(\statespace\) because \(\Vtar\) is proper. Consequently the lower endpoint is exactly \(T_{-}(x) = -\Vtar(x)\), and the flow exists for all
\(t \in (-\Vtar(x),\, +\infty)\).

For a regular value \(c>0\), consider the level set
\(L_{c} = \Vtar^{-1}(c)\), which is a smooth \((m-1)\)-dimensional manifold (by the \emph{Regular Level Set Theorem} in \cite{Lee2012}).
The map
\begin{equation}
\label{eq:apdx_global_prod_struct}
    \Phi : L_{c} \times (-c, \infty) \to \statespace\setminus\{0\},\qquad
    \Phi(p, t) = \phi_{t}(p), \qquad \forall p\in L_{c}, t\in(-c, \infty),
\end{equation}
is a bijection, because every \(x\neq0\) can be written uniquely as
\(\phi_{t}(p)\) with \(p = \phi_{-(\Vtar(x)-c)}(x) \in L_{c}\) and \(t = \Vtar(x)-c\).
The inverse is
\(x \mapsto \bigl( \phi_{-(\Vtar(x)-c)}(x), \Vtar(x)-c \bigr)\).
Both \(\Phi\) and its inverse are \(\continuous{\infty}\) (the flow of a \(\continuous{\infty}\) vector field is \(\continuous{\infty}\)), so \(\Phi\) is a \(\continuous{\infty}\) diffeomorphism.
Thus we obtain a global product structure: every level set of \(\Vtar\) is diffeomorphic to \(L_{c}\), and the punctured domain \(\statespace\setminus\{0\}\) is diffeomorphic to \(L_{c}\times (0,\infty)\) (with diffeomorphism $\Phi$).

\subsection{Morse coordinates near the origin}

\begin{lemma}[Morse Lemma \cite{Milnor1963}]
\label{lemma:morse_lemma_2.2}
    Let $p$ be a non-degenerate critical point for $f$ (where $f$ is a smooth real-valued function), then there is a local coordinate system $(y^1,\cdots,y^n)$ in a neighbourhood $\neighbourhood$ of $p$ with $y^i(p) = 0$ for all $i$ and such that the identity $f = f(p) - (y^1)^2 - \cdots - (y^\lambda)^2 + (y^{\lambda + 1})^2 + \cdots + (y^n)^2$ holds throughout $\neighbourhood$, where $\lambda$ is the index of $f$ at $p$.
\end{lemma}
\begin{corollary}
\label{coro:morse_lemma_2.2}
    There exists an open neighbourhood $\neighbourhood_{0} \subset \statespace$ of \(0\) and a \(\continuous{\infty}\) diffeomorphism
    \begin{equation}
        \varphi : \neighbourhood_{0} \to B_{\delta}(0),\qquad B_{\delta}(0)=\{x|\|x\|<\delta\},
    \end{equation}
    such that \(\Vtar(x) = \|\varphi(x)\|^{2}\) for all \(x \in \neighbourhood_{0}\). And \(\neighbourhood_{0}\) can be chosen such that \(\neighbourhood_{0} \subset \statespace\).
\end{corollary}
\begin{proof}
    Since 0 is a non-degenerate critical point, its index $\lambda$ is 0. Let $p=0$ and $\lambda=0$, by applying \Cref{lemma:morse_lemma_2.2}, we can conclude that there exists a coordinate system $(y^1, \cdots, y^n)$ in a neighbourhood $\neighbourhood$ of 0, such that $f=(y^1)^2 + \cdots + (y^n)^2$ holds throughout $\neighbourhood$.

    We take $\neighbourhood_{0}$ as the intersection $\neighbourhood_{0}=\neighbourhood \cap \statespace$. By definition of open sets, and consider the fact that $0\in\neighbourhood,0\in\statespace\implies0\in\neighbourhood_{0}$, $\neighbourhood_{0}$ is also an open neighbourhood of $p=0$. Thus, we have obtained $\neighbourhood_{0}$ and a coordinate system $(y^1, \cdots, y^n)$, such that the identity $f=(y^1)^2 + \cdots + (y^n)^2$ holds throughout $\neighbourhood_{0}$.

    By taking $\varphi(x)$ to be the coordinate transform, i.e., $\varphi(x)=(y^1, \cdots, y^n)^\intercal$, we have $\Vtar(x)=\|\varphi(x)\|^2$.

    Moreover, by theorem 3.1.1 (Morse Lemma) given in \cite{Nirenberg2001}, there exists such a $\varphi(x)$ that is also $\continuous{\infty}$ smooth. This completes the proof of \Cref{coro:morse_lemma_2.2}.
\end{proof}

\begin{remark}
    Note that \Cref{coro:morse_lemma_2.2} only shows that there exists some $\neighbourhood_{0}\in\statespace$ where $\Vtar(x)=\|\varphi(x)\|^2$ holds, hence letting $\fnwarp=\varphi$ does not complete our proof (which demands $\Vtar(x)=\fnarch{x},\forall x\in\statespace$).
\end{remark}

\subsection{Smoothly map a single level set to a sphere}

Recall that, by \cref{coro:morse_lemma_2.2}, we have $\Vtar(x)=\|\varphi(x)\|^2,\forall x\in\neighbourhood_{0}$, and $\neighbourhood_{0}$ is smoothly deformed by $\varphi$ into an open ball $B_{\delta}(0)$ of radius $\delta$. In other words, $\Vtar(x)<\delta^2,\forall x\in\neighbourhood_{0}$.

Choose \(a = \delta^{2}/4\).  Then \(L_{a} = \Vtar^{-1}(a)\) is contained in the Morse neighbourhood \(\neighbourhood_{0}\).
The diffeomorphism \(\varphi\) maps \(L_{a}\) onto the sphere \(\partial B_{\sqrt{a}}(0)\), consequently
\begin{equation}
    h : L_{a} \to S^{m-1},\qquad h(x) = \frac{\varphi(x)}{\|\varphi(x)\|}
\end{equation}
is a \(\continuous{\infty}\) diffeomorphism onto the unit sphere $S^{m-1}$.

\subsection{Global construction of $\fnwarp$}

Using the product structure \cref{eq:apdx_global_prod_struct} with \(c = a\), define
\begin{subequations}
    \begin{align}
        & F : \statespace\setminus\{0\} \to S^{m-1} \times (0,\infty), \\
        & F(x) = \Bigl( h\bigl( \Phi(x,-(\Vtar(x)-a)) \bigr), \Vtar(x) \Bigr) = \Bigl( h\bigl( \phi_{-(\Vtar(x)-a)}(x) \bigr), \Vtar(x) \Bigr).
    \end{align}
\end{subequations}
Note that \(\phi_{-(\Vtar(x)-a)}(x) \in L_{a}\) because
\(\Vtar(\phi_{-(\Vtar(x)-a)}(x)) = \Vtar(x) - (\Vtar(x)-a) = a\).
The map \(F\) is the composition of the diffeomorphism \(\Phi\) (with the second parameter re‑expressed) and the diffeomorphism \(h\) on the first factor, hence \(F\) is a \(\continuous{\infty}\) diffeomorphism.

\begin{remark}
    In other words, we use $\phi_t$ to diffeomorphically map every level set $L_c$ to $L_a$ via $\phi_{a-c}$, and then smoothly map $L_a$ onto the unit sphere $S^{m-1}$ with $\varphi$ obtained from the Morse coordinates. We also use $\Vtar(x)$ as coordinate along an extra dimension in the product $S^{m-1}\times(0,\infty)$, so that we have a diffeomorphism from $\statespace\setminus\{0\}$ to the product $S^{m-1}\times(0,\infty)$.
\end{remark}

Now compose \(F\) with map \(P : S^{m-1}\times (0,\infty) \to \Real^{m}\setminus\{0\}, P(u,r) = \sqrt{r}\,u\):
\begin{equation}
    \fnwarp_{0}(x) = P(F(x)) = \sqrt{\Vtar(x)} h\bigl( \phi_{-(\Vtar(x)-a)}(x) \bigr), \qquad \forall x \neq 0.
\end{equation}
Set \(\fnwarp_{0}(0)=0\).  For \(x \neq 0\) we immediately obtain
\begin{equation}
    \|\fnwarp_{0}(x)\| = \sqrt{\Vtar(x)} \quad\Longrightarrow\quad \Vtar(x) = \|\fnwarp_{0}(x)\|^{2}, \qquad \forall x \neq 0.
\end{equation}
Note that this is $\continuous{\infty}$ on $\statespace\setminus\{0\}$.

\subsection{Extension to the origin}

To obtain a global diffeomorphism (that additionally includes $x=0$), we modify the normalized gradient field $X(x)$ near the origin so that its flow coincides with the radial flow coming from \(\varphi\).  Choose numbers \(r_1,r_2\) with \(\delta/2<r_1<r_2<\delta\) and a smooth function \(\eta:[0,\infty)\to[0,1]\) such that
\begin{equation}
    \eta(s)=0 \;\; (s\le r_1), \qquad \eta(s)=1 \;\; (s\ge r_2).
\end{equation}
On \(\neighbourhood_0\setminus\{0\}\), define:
\begin{equation}
    X_{\mathrm{rad}}(x) = (D\varphi(x))^{-1}\varphi(x),
\end{equation}
By the chain rule,
\begin{equation}
    \langle \nabla V(x),\, X_{\mathrm{rad}}(x) \rangle
   = \bigl\langle \nabla (\|\varphi(\cdot)\|^{2})(x),\; (D\varphi(x))^{-1}\varphi(x) \bigr\rangle
   = 2\|\varphi(x)\|^{2} = 2V(x).
\end{equation}
Normalize it to
\(Y_{\mathrm{rad}}(x) = X_{\mathrm{rad}}(x)/(2V(x))\), then \(\langle\nabla\Vtar(x), Y_{\mathrm{rad}}\rangle=1\) and \(Y_{\mathrm{rad}}\) is \(\continuous{\infty}\) on $\neighbourhood_{0}\setminus\{0\}$.
Now we define the blended vector field as:
\begin{equation}
    \widetilde{X}(x) =
    \begin{cases}
    \bigl(1-\eta(\|\varphi(x)\|)\bigr)\,Y_{\mathrm{rad}}(x) + \eta(\|\varphi(x)\|)\,X_{\mathrm{grad}}(x), & x\in\neighbourhood_0\setminus\{0\},\\[4pt]
    X_{\mathrm{grad}}(x), & x\in\statespace\setminus\neighbourhood_0 .
    \end{cases}
\end{equation}
Since $\langle\nabla\Vtar(x), Y_{\mathrm{rad}}\rangle=\langle\nabla\Vtar(x), X_{\mathrm{grad}}\rangle=1$, and $\widetilde{X}$ is a linear blend of $Y_{\mathrm{rad}}$ and $X_{\mathrm{grad}}$, we have:
\begin{equation}
    \langle \nabla V(x), \widetilde{X}(x) \rangle = 1, \qquad \forall x\neq 0.
\end{equation}
\(\widetilde{X}\) is \(\continuous{\infty}\) on \(\statespace\setminus\{0\}\) and coincides with \(Y_{\mathrm{rad}}\) on the region
\(\{0<\|\varphi(x)\|\le r_1\}\).

Let \(\widetilde\phi_t\) be the flow of \(\widetilde{X}\).  By the same properness argument as in \Cref{prfsec:global_product_struct},
\(\widetilde\phi_t(x)\) exists for all \(t\in(-V(x),\infty)\) and stays inside \(\statespace\setminus\{0\}\), and $\Vtar(\widetilde\phi_t(x)) = \Vtar(x) + t$.
Because \(\widetilde{X}=Y_{\mathrm{rad}}\) near the origin, the flow there is exactly the radial flow of the Morse chart $\varphi$. In other words, for
\(0<\|\varphi(x)\|\le r_1\), let \(y(t) = \varphi(\widetilde\phi_t(x))\).  Using \(\dot{\widetilde\phi}_t = \widetilde X(\widetilde\phi_t)\)(by definition) and the equality \(\widetilde X = Y_{\text{rad}}\) on the region in question, we obtain
\begin{equation}
\label{eq:dotyt_diffeq}
    \dot y(t) = \frac{y(t)}{2\|y(t)\|^2}, \qquad y(0) = \varphi(x).
\end{equation}
The solution to \cref{eq:dotyt_diffeq} is:
\begin{equation}
    y(t) = \frac{\sqrt{\|\varphi(x)\|^2 + t}}{\|\varphi(x)\|}\, \varphi(x).
\end{equation}
Setting \(t = a - V(x)\) and using \(V(x) = \|\varphi(x)\|^2\) gives
\begin{equation}
    y(a-V(x)) = \frac{\sqrt{a}}{\|\varphi(x)\|}\, \varphi(x) = \sqrt{a}\, \frac{\varphi(x)}{\|\varphi(x)\|},
\end{equation}
hence
\begin{equation}
\label{eq:ext_flow_expr_near_origin}
    \widetilde\phi_{a-V(x)}(x) = \varphi^{-1}\!\left( \sqrt{a}\, \frac{\varphi(x)}{\|\varphi(x)\|} \right).
\end{equation}

Now repeat the construction of \(\Psi_0\) using the new flow.  Let
\begin{equation}
\label{eq:final_construction}
    \Psi(x) = \sqrt{V(x)}\; h\bigl( \widetilde\phi_{a-V(x)}(x) \bigr), \qquad \forall x\neq0,\; \Psi(0)=0 .
\end{equation}
For \(x\) with \(\|\varphi(x)\| \le r_1\), by \cref{eq:ext_flow_expr_near_origin}, we have:
\begin{equation}
    \Psi(x) = \|\varphi(x)\|\; \frac{\varphi(x)}{\|\varphi(x)\|} = \varphi(x), \qquad \forall x\in\{x\in\statespace|V(x)\le a\}.
\end{equation}
Thus \(\Psi\) coincides with the \(\continuous{\infty}\) diffeomorphism \(\varphi\) on the whole ball
\(\{x:\|\varphi(x)\| < r_1\}\).  In particular, \(\Psi\) is \(\continuous{\infty}\) on that neighbourhood, \(\Psi(0)=0\), and
\(D\Psi(0)=D\varphi(0)\) is invertible.

Outside that neighbourhood, \(\Psi\) is the composition of the flow diffeomorphism \(\widetilde\Phi(p,t)=\widetilde\phi_t(p)\),
the diffeomorphism \(h\), and the polar map \(P(u,r)=\sqrt{r}\,u\); hence it is also a
\(\continuous{\infty}\) diffeomorphism. Additionally, by construction in \cref{eq:final_construction}, we have:
\begin{equation}
    \|\Psi(x)\|^2 = V(x), \qquad \forall\,x\in\statespace .
\end{equation}

This completes the proof.
\end{proof}

\section{Computation resource of function fitting experiment}
\label{apdx:compres}

This experiment is conducted on a laptop with Intel i9-13900HX CPU, 32G RAM and NVIDIA RTX4060 GPU (8GB VRAM). Each individual case (out of 16) displayed in \Cref{fig:exp:fit_proper} takes around 1 minute to run. A discrete GPU with at least 2GB VRAM is recommended for running this experiment.

The manually designed scalar fields used in this experiments are saved as 2D images with 1024 pixels in both height and width, and are provided in the supplemental material alongside the code.

\section{Technical details of controller synthesis experiment}
\label{apdx:tech_detail_of_exp}

To ensure correctness of the implementation, we implement the experiment by adapting code publicly released by Wei et al. (\url{https://github.com/shiqingw/Neural-Lyapunov-Uncertainties}, \texttt{eg1\_inverted\_pendulum\_2d/inv\_pend\_2d\_sum4\_nn\_controller.py}). A full list of all relevent hyperparameters is given in \Cref{tab:hyperparam_exp2023}.

\begin{table}[!ht]
  \caption{Hyperparameters of controller synthesis experiment. Hyperparameters not listed here are kept the same as in the inverted pendulum experiment in \cite{Wei2023NeuralLC} (this includes the learning rate and batch size).}
  \label{tab:hyperparam_exp2023}
  \centering
  \begin{tabular}{ccc}
    \toprule
        Name & Value & Remarks \\
    \midrule
        roa\_adaptive\_level\_multiplier & False & Whether to decrease RoA level multiplier over time \\
        cutoff\_radius & 0.1 & Radius of initial safe set \\
    \bottomrule
  \end{tabular}
\end{table}

Note that, the authors of \cite{Wei2023NeuralLC} applied a loose saturation function to the output of the controller only to limit its capability and to demonstrate that their method could indeed gradually expand the RoA of the equilibrium \cite{Wei2023NeuralLC}. In our experiment, we removed this function and use the output of the controller network as-is. Also, \cite{Wei2023NeuralLC} took uncertainty of the system dynamics into account, and used both an accurate "true" dynamics, and a nominal dynamics that is allowed to be slightly inaccurate in their algorithm. We consider a simplified case where the nominal system dynamics is accurate by setting the nominal system to be exactly the same as the true system.

It should also be noted that this method starts by first synthesizing an LQR \cite{Lewis2012} controller for linearization of the full nonlinear system dynamics at $x=0$. For this, we also give the linearization of the system described in \cref{eq:exp_sys} as:
\begin{equation}
    \begin{cases}
        \dot{x_1} = x_2 \\
        \\
        \dot{x_2} = \dfrac{\pi}{2}x_1 + 10u_1 - \dfrac{1}{2}u_2 .
    \end{cases}
\end{equation}
Note that this method is also effective for uncertain nonlinear systems, this linearization is only used as part of the method, see \cite{Wei2023NeuralLC} for further details.

This experiment is performed on a laptop with Intel i9-13900HX CPU, 32G RAM for around 40 minutes. This experiment is performed on CPU without using CUDA.

\section{Controller synthesis experiment with an alternative method}
\label{apdx:exp2025}

This experiment follows the same design as in \Cref{sec:exp:controller_synthesis}, but with another neural Lyapunov control method proposed in \cite{rnc-nyu-2025}. The dynamic system being considered is:
\begin{equation}
\label{eq:exp_sys_2025}
    \begin{cases}
        \dot{x}_1 = x_2 \\
        \\
        \dot{x}_2 = \dfrac{1}{2} \sin(\pi x_1) + (10\exp(-2x_1^2-2x_2^2)-5)u_1 + 2u_2
    \end{cases},
\end{equation}
where $x=(x_1\;x_2)^\intercal,x\in\statespace=(-1,+1)^2$ is the system state, $u=(u_1\;u_2)^\intercal\in(-5,+5)^2$ is the input.

To ensure correctness, we implement this experiment by adapting code publicly released by Wei et al. (\url{https://github.com/shiqingw/Robust-Lyapunov-Prob}, MIT License). A full list of all relevent hyperparameters is given in \Cref{tab:hyperparam_exp2025}. This experiment is conducted on a laptop with Intel i9-13900HX CPU, 32G RAM and NVIDIA RTX4060 GPU (8GB VRAM) for 200 epochs, with 10000 samples in each epoch, following the practice of the authors. Each run takes around 5 minutes. The results are shown in \Cref{fig:exp_fin_all_2025}.

\begin{figure}[!ht]
    \centering
    \includegraphics[width=1\linewidth]{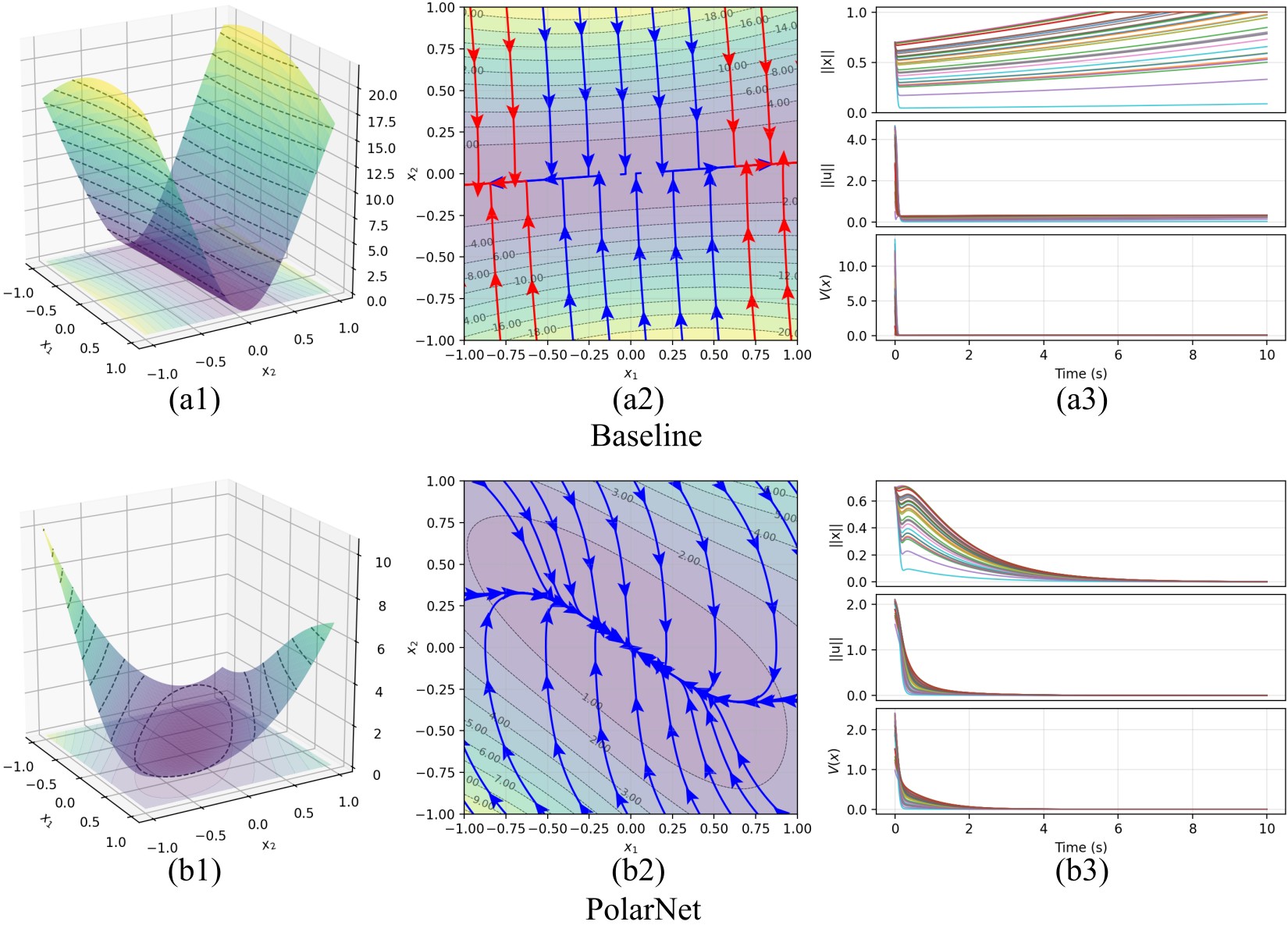}
    \caption{Controller synthesis result with different network architectures: (a1), (b1): Lyapunov function. (a2), (b2): Phase portrait, unstable trajectories are marked in red. (a3), (b3): 36 numerically simulated trajectories. (a): Using the network architecture as in \cite{Wei2023NeuralLC}. (b): Using \ourarch. $\Vtar(x)$ and $\|x\|$ trajectories are collected by uniformly sample $x_0$ on a circle of 0.7 radius centering at the origin. Notice how states evolve along descending directions of $\netV$ for $x\in(-1,+1)\times(-1,+1)$ in (a) but do not converge to the origin, as analyzed in \Cref{sec:limit_of_exiting_methods}.}
    \label{fig:exp_fin_all_2025}
\end{figure}

\begin{table}[!ht]
  \caption{Hyperparameters of controller synthesis experiment with alternative method \cite{rnc-nyu-2025}. Hyperparameters not listed here are kept the same as in the inverted pendulum experiment (the version without distinct nominal and true systems) in \cite{rnc-nyu-2025}.}
  \label{tab:hyperparam_exp2025}
  \centering
  \begin{tabular}{ccc}
    \toprule
        Name & Value & Remarks \\
    \midrule
        lyapunov\_lr & 5e-4 & learning rate of Lyapunov network $\netV$ \\
        lyapunov\_wd & 1e-5 & weight decay (parameter of Adam optimizer) of Lyapunov network $\netV$ \\
        controller\_lr & 2e-3 & learning rate of controller network $\netu$ \\
        controller\_wd & 1e-5 & weight decay (parameter of Adam optimizer) of controller network $\netu$ \\
        cutoff\_radius & 0.1 \\
    \bottomrule
  \end{tabular}
\end{table}

Note that we are adapting the code provided in \texttt{eg1\_inverted\_pendulum/train\_nominal.py} of the repository (\url{https://github.com/shiqingw/Robust-Lyapunov-Prob}), and our statespace $\statespace$ has values in $(-1,+1)$ in both dimensions, rather than $(-\pi,+\pi)$. We limit the controller output to $(-5,+5)$ following the practice of in \cite{rnc-nyu-2025}.

\ifx\paperversion\paperversionpreprint
\else
    \newpage
    \input{checklist.tex}
\fi

\end{document}